\documentclass[10pt, conference, a4paper]{IEEEtran}
\usepackage{amssymb}
\usepackage{amsfonts}
\usepackage{epsf}
\usepackage{epsfig}
\usepackage{epstopdf}
\usepackage{array}
\usepackage{amsmath}
\usepackage{url}
\usepackage[algoruled, linesnumbered]{algorithm2e}
\usepackage{amsthm}
\usepackage{xifthen}
\usepackage{mathrsfs}
\usepackage[mathscr]{euscript}
\usepackage[sans]{dsfont}
\usepackage[dvipsnames]{xcolor}
\usepackage{graphicx}
\usepackage{caption}
\usepackage{subcaption}

\usepackage{epsf}
\usepackage{epsfig}

\usepackage{array}

\usepackage{times}
\usepackage{graphicx}
\usepackage{amsmath}
\usepackage{amssymb}
\usepackage{amsxtra}
\usepackage{here}
\usepackage{rawfonts}
\usepackage{times}
\usepackage{url}
\usepackage{cite}
\usepackage{amsthm}
\usepackage{graphicx}

\usepackage{cite} 

\usepackage[T1]{fontenc}
\usepackage[utf8]{inputenc}
\usepackage{authblk}
\IEEEoverridecommandlockouts
\newtheorem{prop}{Proposition}

\newtheorem{remark}{Remark}

\theoremstyle{definition}
\newtheorem{defn}{Definition}

\ifCLASSINFOpdf
\else
\fi

\newlength{\aligntop}
\newlength{\alignbot}
\makeatletter

\begin{document}
%
\title{Cloud Radio Access meets Heterogeneous Small Cell Networks: A Cognitive Hierarchy Perspective}
\author[D. Kopta et al.]
{Nof Abuzainab and Walid Saad
	\\
	Wireless@VT, Department of Electrical and Computer Engineering, Virginia Tech, Blacksburg, VA, USA\\
	Emails:\{nof, walids\}@vt.edu\\ 
	\vspace{-6.8ex}\thanks{This research was supported by the U.S. Office of Naval Research (ONR) under Grant N00014-15-1-2709}}

\maketitle


%
\IEEEpeerreviewmaketitle

\begin{abstract}
In this paper, the problem of distributed power allocation is considered for the downlink of a cloud radio access network (CRAN) that is coexisting with a heterogeneous network. In this multi-tier system, the heterogeneous network base stations (BSs) as well as the CRAN remote radio heads seek to choose their optimal power to maximize their users' rates. The problem is formulated as a noncooperative game in which the players are the CRAN's cloud and the BSs. Given the difference of capabilities between the CRAN and the various BSs, the game is cast within the framework of cognitive hierarchy theory. In this framework, players are organized in a hierarchy in such a way that a player can choose its strategy while considering players of only similar or lower hierarchies. Using such a hierarchical design, one can reduce the interference caused by the CRAN and high-powered base stations on low-powered BSs. For this game, the properties of the Nash equilibrium and the cognitive hierarchy equilibrium are analyzed. Simulation results show that the proposed cognitive hierarchy model yields significant performance gains, in terms of the total rate, reaching up to twice the rate achieved by a classical noncooperative game's Nash equilibrium.
\end{abstract}
\section{Introduction}

The deployment of a cloud-based radio access network (CRAN) is expected to be an integral part of 5G networks \cite{cloudbenefits2}. In CRAN, remote radio heads (RRHs), which are small radio antennas responsible for transmission, cooperate to serve their wireless users. To achieve such coordination, the RRHs are connected through a fronthaul link to a control unit (CU) that performs efficient signal processing. However, in CRAN, the fronthaul links are often of finite capacity \cite{hetcloud}, and thus, the number of RRHs connected to the CU will be limited. Thus, the CRAN will need to co-exist with the small cell base stations (BSs) of an existing heterogeneous cellular network (HetNet). In particular, in the downlink, the presence of a large number of RRHs can increase the interference on HetNet users. Thus, it is imperative to devise new solutions for optimized interference management in such a hybrid CRAN-HetNet architecture \cite{hetcloud}.

In the existing literature, the problem of downlink interference mitigation has been considered mainly within the context of HetNets \cite{hetint1,hetint2,hetint3}. The objective is to reduce to the interference caused by the high-powered macro BSs on the low-powered femto and pico BSs. These techniques range from power control \cite{hetint1} to the allocation of frequency and time resources such as in \cite{hetint2} and \cite{hetint3}. More recent works such as in \cite{cloud0} considered interference mitigation for a CRAN using MIMO techniques. The authors in \cite{cloud1} studied the challenge of limited fronthaul capacity in CRAN, and derived the optimal power allocations that maximize the system's sum-rates under a constrained fronthaul. In \cite{cloud2}, the authors proposed a centralized resource allocation technique for the downlink of a CRAN coexisting with a macro BS. However, this existing body of work is either focused on CRAN or HetNets, in isolation, or in a very restricted setting. Thus, none of these works addresses the problem of interference management in hybrid CRAN/small cell architectures.

The main contribution of this paper is to introduce a novel framework for power control and interference management that can ensure an efficient HetNet and CRAN co-existence. We formulate the problem as a noncooperative game between the CRAN and HetNet's BSs. In this game, each base station seeks to choose its optimal power allocation to maximize the sum-rate of its users. Given the heterogeneity of the system, we cast the game within the framework of \emph{cognitive hierarchy theory} \cite{Camerer_acognitive} in which players can be organized into multi-tier hierarchies depending on their capabilities, unlike in classical Stackelberg games that allow only for a two level pre-determined hierarchy \cite{stackelberg1}. For our model, the hierarchy pertains to the fact that computationally capable and high-powered BSs, such as the CRAN or the macrocell BSs can have a higher hierarchy in the network than smaller BSs, such as femtocells. Such a hierarchical organization eventually allows reducing the interference on low-powered BSs as they no longer need to consider the actions of the devices with higher power capabilities. In contrast, by being at the top of the hierarchy, high-powered BSs such as the CRAN and macrocell BSs will be more conservative in their power control decisions. For both the conventional game and the cognitive hierarchy game, we characterize the equilibrium solutions and study their properties. Simulation results show that the proposed cognitive hierarchy model yields a two-fold improvement in the total system rate, when compared with a classical game-theoretic model. The results also show considerable interference reduction on the low-powered BSs.

The paper is organized as follows: Section II presents the system model and the power allocation problem. Section III presents the game theoretic formulation of the problem as well as the Nash equilibrium and the cognitive hierarchy solutions respectively. Section IV presents the simulation results. Finally, conclusions are drawn in section V.

\vspace{-0.2 cm}
\section{System Model}

Consider the downlink of a CRAN composed of a set $\mathcal{K}_c$ of remote radio heads (RRH) serving cooperatively a set $\mathcal{N}_c$ of users. All RRHs are connected to a central cloud-based control unit (CU) via finite capacity fronthaul link. The CU is responsible for the resource allocation and the signal processing of the RRHs, while the RRHs are responsible only for data transmission. 
 

Due to the finite capacity of the fronthaul, there is a maximum number of RRHs that could be connected to the CU. Hence, the CRAN must co-exist with a HetNet composed of a set $\mathcal{M}$ of macro BSs, a set $\mathcal{P}$ of pico BSs, and a set $\mathcal{F}$ femto BSs. Each BS $i$ $ \in \mathcal{H}=\mathcal{M} \cup \mathcal{P} \cup \mathcal{F}$ has a set $\mathcal{N}_i$ of users to serve. Each user $j \in \cup_{i \in \mathcal{H}} \mathcal{N}_i \cup \mathcal{N}_c$ is located at a distance $d_{ij}$ from BS $i \in \mathcal{H}$ or from RRH $i \in \mathcal{K}_c$.
 OFDMA transmissions of $L$ subcarriers are also assumed at each BS in $\mathcal{H}$ and RRH in the CRAN.
The channel gain $h_{ijk}$ between each BS/RRH $i$ and user $j$ over subcarrier $k$ is distributed according to an independent Rayleigh fading. The received power of a user $j$ over subcarrier $k$ is
$P_{r_{ijk}}=|h_{ijk}|^2p_{ik}d_{ij}^{-\alpha}$ where $\alpha$ is the path loss exponent and $p_{ik}$ is the transmission power of BS/RRH $i$ over subcarrier $k$. Additive white Gaussian noise of variance $\sigma^2$ is present at each receiver. In our model, each BS/RRH $i$ can transmit with a maximum power $P_{i,\textrm{max}}$, and each mobile user has a single receive antenna.

In the CRAN, the CU assigns subcarrier $k$ to a single user in $\mathcal{N}_c$ and all RRHs transmit simultaneously to that user. The CU uses beamforming to coordinate the transmissions of the RRHs over subcarrier $k$.
Thus, the received rate of user $j \in \mathcal{N}_c$ over subcarrier $k$ is given by the multiple input single output capacity with power constraint on each antenna as in \cite{miso}:\vspace{-0.1 cm}\begin{equation}
R_{jk}(\boldsymbol{p}_{k})=\frac{W}{L}\log\bigg(1+\frac{\Big(\sum_{i=1}^{|\mathcal{K}_c|} |h_{ijk}|\sqrt{p_{ik}d_{ij}^{-\alpha}}\Big)^2}{\sigma^2+\sum_{l \in \mathcal{H}}|h_{ljk}|^2p_{lk}d_{lj}^{-\alpha}}\bigg),\label{ratecloud}
 \end{equation}\normalsize where $\boldsymbol{p}_k$ is the vector of power allocations of all RRHs over subcarrier $k$ and $W$ is the transmission bandwidth. It has been shown in \cite{miso} that the optimal signaling to achieve this capacity is to use beamforming weight vector $\boldsymbol{v}$ where each entry $v_i$ is the weight of RRH $i$ and is given by:\vspace{-0.25 cm}
\begin{equation}\small
v_i=\frac{h^*_{ijk}\sqrt{p_{ik}d_{ij}^{-\alpha}}}{|h_{ijk}|\sum_{l \in \mathcal{K}_c}\sqrt{p_{lk}d_{lj}^{-\alpha}}}.
\vspace{-0.1 cm}
\end{equation}\normalsize Based on this transmission scheme and the rate equation in (\ref{ratecloud}), and in order to ensure fairness, the CU assigns subcarrier $k$ to user $j^* \in \mathcal{N}_c$ such that:
\vspace{-0.15 cm}
\begin{equation}\small
j^*=\arg \max_j \frac{\sum_{i=1}^{|\mathcal{K}_c|} |h_{ijk}|\sqrt{d_{ij}^{-\alpha}}}{\bar{R}_{jk}}, \label{fair}
\vspace{-0.1 cm}\end{equation}\normalsize assuming that transmissions occur within time frames and $\bar{R}_{jk}$ is the average rate over the previous time slots. Since each BS assigns each subcarrier to one of its users, hereinafter, we drop the user index $j$ for brevity. The CU finds the optimal power allocations $\boldsymbol{p}$ for all RRHs that maximize the sum of rates over $L$ subcarriers:
\vspace{-0.2 cm} \begin{eqnarray}
&&\max_{\boldsymbol{p}}\sum_{k=1}^{L} R_{k}(\boldsymbol{p}_{k}),\hspace{0.2 cm}\textrm{s.t.} \sum_{k=1}^L p_{ik} \leq P_{i,\textrm{max}}, \hspace{0.2 cm} \forall i \in \mathcal{K}_c.
\end{eqnarray}
 As for the BSs in $\mathcal{H}$, the received rate at each user $j$ in $\mathcal{N}_i$ from BS $i$ transmission will be: \begin{equation}
 R_{ijk}(p_{ik})=\frac{W}{L}\log\Big(1+ \frac{|h_{ijk}|^2p_{ik}d_{ij}^{-\alpha}}{\sigma^2+ \sum_{l \in \mathcal{H} \cup \mathcal{K}_c, l \neq i}|h_{ljk}|^2p_{lk}d_{lj}^{-\alpha}}\Big),\label{BSrate}
\end{equation}
Based on (\ref{BSrate}), BS $i$ allocates subcarrier $k$ to user $j^*$ in $\mathcal{N}_i$ such that:
$j^*=\textrm{arg} \max_j \frac{|h_{ijk}|}{\bar{R}_{jk}}.$ Each BS $i$ is required to find its transmission powers $p_{ik}$ so that the sum of rates over all subcarriers is maximized:\begin{eqnarray}
&&\max_{p_{ik}}\sum_{k=1}^{L} R_{ik}(p_{ik}),\hspace{0.2 cm}\textrm{s.t.} \sum_{k=1}^L p_{ik} \leq P_{i,\textrm{max}}. \label{BSopt} 
\end{eqnarray}
The BSs in the HetNet operate independently of each other and of the CRAN. Hence, the CU and the BSs must find their optimal transmission powers in a distributed manner. However, the RRHs and the BSs will mutually interfere which, in turn, affects their users' received rates. This motivates the adoption of a game-theoretic approach, as explained next.

\section{Distributed Power Allocation: Game Theory and Cognitive Hierarchy}
This problem is formulated as a static noncooperative game defined by the triplet $(\mathcal{P},(\mathcal{S}_i)_{i \in \mathcal{P}}, (\mathcal{U}_i)_{i \in \mathcal{P}})$ where the $\mathcal{P}=\mathcal{H}\cup \{CU\}$ is the set of players. The strategy vector $\boldsymbol{p}_i$ of each player $i$ is to choose its transmission power values over all the subcarriers. Hence, the strategy set for BS $i$ in $\mathcal{H}$ is: $\mathcal{S}_i=\{\boldsymbol{p}_i \in [0,P_{i,\textrm{max}}]^{L} \hspace{0.1 cm} \textrm{s.t.} \hspace{0.1 cm} \sum_{k=1}^Lp_{ik}=P_{i,\textrm{max}} \}$, and the strategy set for the CU in the CRAN is $\mathcal{S}_C=\{\boldsymbol{p_c} \in \prod_{i \in \mathcal{K}_c}[0,P_{i,\textrm{max}}] \hspace{0.1 cm} \textrm{s.t.} \hspace{0.1 cm} \sum_{k=1}^Lp_{ik}=P_{i,\textrm{max}} \forall i \in \mathcal{K}_c\}$. The utility for each BS $i$ will be:
\begin{eqnarray}
 U_i(\boldsymbol{p}_i,\boldsymbol{p}_{-i})=\sum_{k=1}^LR_{ik}(p_{ik})
\end{eqnarray}\normalsize where $\mathcal{V}=\mathcal{H} \cup \mathcal{K}_c$ and $d_{ik}$ is the distance between BS $i$ and the user assigned to subcarrier $k$. Similarly, the CU's utility is:
\begin{eqnarray} U_C(\boldsymbol{p}_c,\boldsymbol{p}_{-c})=\sum_{k=1}^LR_k(\boldsymbol{p}_k) 
 \label{ucloud}
\end{eqnarray}\normalsize 

\subsection{Nash Equilibrium Solution}
The conventional solution for this formulated game is the so-called Nash equilibrium \cite{gametheory}. In particular, a strategy profile $\boldsymbol{p}^*$ is said to constitute a Nash equilibrium (NE) if:
\begin{equation}\small
U_i(\boldsymbol{p}^*_i,\boldsymbol{p}^*_{-i}) \geq U_i(\boldsymbol{p}_i,\boldsymbol{p}^*_{-i}) \hspace{0.4 cm} \forall i \in \mathcal{P}, \boldsymbol{p}_i \in \mathcal{S}_i.
\end{equation}where $\boldsymbol{p}_{-i}$ is the vector of strategies of all players except $i$. To derive the NE for this game, we first determine the best response strategy for each player $i \in \mathcal{P}$. For each BS $i \in \mathcal{H}$, the best response is the well known water-filling solution \cite{waterfilling}.
\begin{remark}
\emph{Given a fixed power policy $\boldsymbol{p}_{-i}$, the optimal strategy (best response) of BS $i$ is given by:
\begin{equation}
\small
p_{ik}=\Big(\frac{W}{L\mu_i}-\frac{1}{c_{ik}}\Big)^+,
\end{equation}
where $c_{ik}$ is given by:
\begin{equation}
\small
c_{ik}=\frac{|h_{ik}|^2d_{ik}^{-\alpha}}{\sigma^2+ \sum_{l \in \mathcal{V}, l \neq i}|h_{lk}|^2p_{lk}d_{lk}^{-\alpha}},
\end{equation}
and $\mu_i$ satisfies: $\small
\sum_{k=1}^{L}\Big(\frac{W}{L\mu_i}-\frac{1}{c_{ik}}\Big)^+=P_{i,\textrm{max}}$}.
\end{remark}
This follows from the fact that utility function is concave in $\boldsymbol{p}_i$ since it is a sum of log functions such that each term is a log function of $p_{ik}$ (for $1 \leq k \leq L$). Also, the constraint function of $p_{ik}$ in (\ref{BSopt}) is a linear function of $p_{ik}$ (for $1 \leq k \leq L$). Thus, the optimal solution is found using the KKT conditions.

\begin{prop}
\emph{The best response of the CU is the solution of the nonlinear equations:
\begin{eqnarray}
&&\forall i \in \mathcal{K}_c, 1 \leq k \leq L, \nonumber\\
&&\frac{W\Big(c^2_{ik}-\frac{c_{ik}}{\sqrt{p_{ik}}}\sum_{l \neq i, l \in \mathcal{K}_c}c_{lk}\sqrt{p_{lk}}\Big)}{L\Big(e_k+\Big(\sum_{i=1}^{|\mathcal{K}_c|} |h_{ik}|\sqrt{p_{ik}d_{ik}^{-\alpha}}\Big)^2\Big)}-\mu_i=0,\nonumber\\
&&\mu_i\Big(\sum_{k=1}^Lp_{ik}-P_{i,\textrm{max}}\Big)=0, \hspace{0.1 cm} \forall i \in \mathcal{K}_c, 
\end{eqnarray}  
where $c_{ik}$ and $e_k$ are given by respectively:}

$c_{ik}= |h_{ik}|\sqrt{d_{ik}^{-\alpha}},$

$e_k=\sigma^2+ \sum_{l \in \mathcal{H}, }|h_{lk}|^2p_{lk}d_{lk}^{-\alpha}$.
\vspace{-0.2 cm}
\end{prop}


\begin{proof}
We first show that
$\log\bigg(1+\frac{\Big(\sum_{i=1}^{|\mathcal{K}_c|} |h_{ik}|\sqrt{p_{ik}d_{ik}^{-\alpha}}\Big)^2}{\sigma^2+I}\bigg)$ is concave in the variables $p_{ik}$ ($1 \leq i \leq |\mathcal{K}_c|$) by showing that the function $\Big(\sum_{i=1}^{|\mathcal{K}_c|} |h_{ik}|\sqrt{p_{ik}d_{ik}^{-\alpha}}\Big)^2$ is concave. This function can be written as:\small\begin{eqnarray}
\Big(\sum_{i=1}^{|\mathcal{K}_c|} |h_{ik}|\sqrt{p_{ik}d_{ik}^{-\alpha}}\Big)^2&=& \Big(\sum_{i=1}^{|\mathcal{K}_c|}|h_{ik}|^2p_{ik}d_{ik}^{-\alpha}
\nonumber\\&&+ \sum_{i \neq r}|h_{ik}||h_{rk}|\sqrt{p_{ik}p_{rk}d_{ik}^{-\alpha}d_{rj}^{-\alpha}}\Big).\nonumber \vspace{-0.2 cm}
\vspace{-0.2 cm}\end{eqnarray}\normalsize   In the above equation, the first sum  is a linear sum of $p_{ik}$ for $1 \leq i \leq|\mathcal{K}_c|$. Also, each term of the second sum can be easily seen to be concave in $p_{ik}$ and $p_{rk}$ since in general the function $\sqrt{xy}$ is concave in $x$ and $y$. Thus, the function $\Big(\sum_{i=1}^{|\mathcal{K}_c|} |h_{ik}|\sqrt{p_{ik}d_{ik}^{-\alpha}}\Big)^2$ is a sum of concave functions in $p_{ik}$ for $1 \leq i \leq|\mathcal{K}_c|$ and is therefore concave.
Also in general, the function $\log(1+x)$ is increasing concave function of $x$. Thus, the function $\log\bigg(1+\frac{\Big(\sum_{i=1}^{|\mathcal{K}_c|} |h_{ik}|\sqrt{p_{ik}d_{ik}^{-\alpha}}\Big)^2}{\sigma^2+I}\bigg)$ is an increasing concave function of $\Big(\sum_{i=1}^{|\mathcal{K}_c|} |h_{ik}|\sqrt{p_{ik}d_{ik}^{-\alpha}}\Big)^2$. Thus, it is concave using the general property that a monotonic concave transformation of a concave function is concave. It follows that the utility function of the CU in (\ref{ucloud}) is also concave since it is the sum of concave functions in $p_{ik}$, $1 \leq i \leq|\mathcal{K}_c|$. Thus, the optimal solution can be found using KKT conditions.
\end{proof}\vspace{-0.25 cm}The existence of the NE is presented next.

\begin{prop}
\emph{There exists at least one pure Nash equilibrium for the heterogeneous cloud power allocation game.}
\end{prop}
\begin{proof}
This follows directly from \cite[Proposition 20.3]{purenash}.
\end{proof}
At the NE, the CU or each BS considers the strategies made by all other devices when choosing its own strategy. This can be detrimental to low-powered BSs who will need to respond to a high level of interference caused by the CRAN and macro BSs, and hence, will choose a more conservative power allocation. To overcome this challenge, we will next develop a new game-theoretic model based on \emph{cognitive hierarchy theory} \cite{Camerer_acognitive}.

\subsection{Cognitive Hierarchy Approach}

 Under cognitive hierarchy (CH), devices are organized into a hierarchy such that the CRAN and macro BSs are placed at the higher hierarchy levels, while low-powered BSs are placed at the lower hierarchy levels. Each player assumes that the other players belong to equal or lower hierarchy levels when selecting its strategy. This allows low-powered BSs to choose less conservative actions than in a conventional NE. We adopt the Poisson CH model in \cite{Camerer_acognitive} with the following considerations. First, players are distributed across hierarchy levels according to a Poisson distribution $f$ of rate $\tau$. Moreover, a player at level $k$ knows the true distributions of players which are at lower levels. This is known as the overconfidence assumption and it is a realistic assumption for cases in which players are humans for example. Using hierarchical approaches for power allocation in HetNets has been earlier investigated using Stackelberg games \cite{stackelberg1}. Stackelberg games, however, rely on only a two level hierarchy composed of leaders and followers, while CH allows for multiple levels of hierarchy, which is more suitable to our setting. Further, in Stackelberg games, both leaders and followers consider the actions of each others when choosing their actions. This may not be very effective for reducing interference on low-powered BSs.

%
%
%
%


 In the HetNet tier, the maximum transmit power of a BS is dependent on its type. For instance, a larger BS has a higher maximum transmit power, and thus, it can cause higher interference on smaller BSs. To mitigate this high interference, we group the BSs based on their type, using the CH level approach. In particular, femtocell BSs, who have the lowest power, are assumed to belong to level 1 while pico BSs belong to level 2, macro BSs belong to level 3, and the CRAN belongs to the last level. 
In order to reduce the interference among BSs of the same type, the overconfidence assumption of the Poisson CH model is modified so that BSs consider the actions of BSs of the same type. However, in CH, each BS takes its action based on its beliefs and does not learn the power values of the BSs of the same type. Thus, we assume that each player at level $k$ believes that all other players of the same type act using the same strategy. Consequently, the utility of BS $i$ at level $m$ is:
\begin{equation}
U_{im}(\boldsymbol{p}_i,\boldsymbol{p}_{-i})=\frac{W}{L}\sum_{k =1 }^{L}\log\Big(1+ \frac{|h_{ik}|^2p_{ik}d_{ik}^{-\alpha}}{\sigma^2+\mathbb{E}_{g_m}[I(\boldsymbol{p}_{-i})])} \bigg),\label{BSCH}
\vspace{-0.3 cm}
\end{equation}
where
\begin{eqnarray} 
\vspace{-0.6 cm}
\mathbb{E}_{g_m}[I(\boldsymbol{p}_{-i})]&=& g_m(m)\sum_{l \in \mathcal{V}, l \neq i}|h_{lk}|^2p_{ik}d_{lk}^{-\alpha}\nonumber\\
&&+\sum_{l \in \mathcal{V}, l \neq i}\sum_{h=0}^{m-1}g_m(h)|h_{lk}|^2p_{lk}(h)d_{lk}^{-\alpha},\nonumber
\vspace{-0.6 cm}\end{eqnarray}\normalsize and $g_m(h)$ is the proportion of players at level $h$ ($1 \leq h \leq m$). It is given by: $g_m(h)=\frac{f(h)}{\sum_{i=0}^{m}f(i)}$.

\vspace{0.2 cm}
For the CRAN, we assume that the CU is aware that there are no other players at the same level. This can be achieved by selecting a proper value of the Poisson distribution parameter $\tau$ that satisfies $f(4)\approx 0$ in the CH model. One suitable value here is: $\tau=1$.
The utility of the CU at level 4 is thus given by:\vspace{-0.2 cm}\begin{equation}\small 
U_{C}(\boldsymbol{p}_C,\boldsymbol{p}_{-C})=\frac{W}{L}\sum_{k =1 }^{L}\log\bigg(1+\frac{\Big(\sum_{i=1}^{|\mathcal{K}_c|} |h_{ik}|\sqrt{p_{ik}d_{ik}^{-\alpha}}\Big)^2}{\sigma^2+\mathbb{E}_{g_4}[I(\boldsymbol{p}_{-i})]} \bigg),\vspace{-0.2 cm}
\end{equation}
where
\vspace{-0.2 cm}
\begin{equation} 
\mathbb{E}_{g_4}[I(\boldsymbol{p}_{-i})]=\sum_{l =0}^{\mathcal{H}}\sum_{h=0}^{3}g_4(h)|h_{lk}|^2p_{lk}(h)d_{lk}^{-\alpha}.\vspace{-0.1 cm}
\end{equation}\normalsize

 The equilibrium for the CH case is defined next.
\begin{defn}
 A strategy profile $\boldsymbol{p}^*$ is said to constitute a \emph{cognitive hierarchy equilibrium (CHE)} iff:
 \vspace{-0.1 cm}
\begin{equation}
\boldsymbol{p}^*_j= \arg \max_{\boldsymbol{p}_j \in \mathcal{S}_j} U_{jm}(\boldsymbol{p}_j,\boldsymbol{p}^*_{-j}) \hspace{0.1 cm} \forall j \in \mathcal{P},
\vspace{-0.1 cm}
\end{equation}
where $m$ is the CH level of player $j$. For this CH game, the equilibrium strategy for each BS and for the CU are given respectively by the following two propositions.
\end{defn}
\begin{prop}\emph{The CHE strategy for each BS $i \in \mathcal{H}$ at CH level $m$ is achieved by solving the following system of equations:
	\begin{eqnarray}
	&&\forall 1\leq k \leq L, \nonumber\\
	&&\frac{W c_{ik} I_{m-1}}{L(I_{m-1}+D_mp^*_{ik})(I_{m-1}+(c_{ik}+D_m)p^*_{ik})}-\mu_i=0, \nonumber\\
	&&\mu_i(\sum_{k=1}^Lp^*_{ik}-P_{i,\textrm{max}})=0,
\label{BSCHE}\vspace{-0.6 cm}	\end{eqnarray}\vspace{-0.2 cm}}

	\emph{where} $c_{ik}=|h_{ik}|^2d_{ik}^{-\alpha}$,
$D_m=g_m(m)\sum_{l \in \mathcal{V}, l \neq i}|h_{lk}|^2d_{lk}^{-\alpha},$
	
	$I_{m-1}=\sigma^2+ \sum_{l \in \mathcal{V}, l \neq i}\sum_{h=0}^{m-1}g_m(h)|h_{lk}|^2p^*_{lk}(h)d_{lk}^{-\alpha}.$
\end{prop}
\begin{proof}
	It can be easily shown that the function $\log(1+\frac{ax}{bx+c})$ is concave in $x$ $\forall a\geq 0,b \geq 0,c \geq 0,x \geq 0$. Thus, the term $\log\Big(1+ \frac{|h_{ik}|^2p_{ik}d_{ik}^{-\alpha}}{\sigma^2+\mathbb{E}_{g_m}[I(\boldsymbol{p}_{-i})])} \bigg)$ is concave in $p_{ik}$. Thus, the utility in (\ref{BSCH}) is a sum of concave functions.
\end{proof}
\begin{prop}
		\emph{The CHE strategy of the CU is the solution of the nonlinear equations:
		\begin{eqnarray}
		&&\forall i \in \mathcal{K}_c, 1 \leq k \leq L, \hspace{5 cm}\nonumber\\
		&& \frac{W\Big(c^2_{ik}-\frac{c_{ik}}{\sqrt{p^*_{ik}}}\sum_{l \neq i, l \in \mathcal{K}_c}c_{lk}\sqrt{p^*_{lk}}\Big)}{L\Big(e_k+\Big(\sum_{i=1}^{|\mathcal{K}_c|} |h_{ik}|\sqrt{p^*_{ik}d_{ik}^{-\alpha}}\Big)^2\Big)}-\mu_i=0,\nonumber\\
		&&\mu_i\Big(\sum_{k=1}^Lp^*_{ik}-P_{i,\textrm{max}}\Big)=0, \hspace{0.1 cm} \forall i \in \mathcal{K}_c, \label{CUCHE} \vspace{-0.2 cm}
		\end{eqnarray} 
		\vspace{-0.1 cm} 
		where $c_{ik}$ and $e_k$ are given by respectively:}
\end{prop}
	
$c_{ik}= {|h_{ik}|\sqrt{d_{ik}^{-\alpha}}},$

$e_k=\sigma^2+\sum_{l =0}^{\mathcal{H}}\sum_{h=0}^{3}g_4(h)|h_{lk}|^2p^*_{lk}(h)d_{lk}^{-\alpha}.$
\begin{proof}
	The utility of the CU in the CH case is still the same function of $p_{ik}$ as in (\ref{ucloud}). Hence, the function is concave and the above equations are obtained by applying the KKT conditions. 
\end{proof}
\vspace{-0.3 cm}
Based on (\ref{BSCHE}) and (\ref{CUCHE}), each player at level $m$ finds the level $h$ CHE power $p^*_{lk}(h)$ for each other player $l$ and for all $h < m$ in order to find its CHE strategy.

\section{Simulation Results}
To assess the performance achieved by both NE and CH, the following values are considered: $|\mathcal{K}_c|=40$, $|\mathcal{N}_c|=70$,  $|\mathcal{K}_m|=5$, $|\mathcal{N}_i|=25$ $\forall i \in \mathcal{K}_m$, $|\mathcal{K}_p|=5$, $|\mathcal{N}_i|=15$ $\forall i \in \mathcal{K}_p$, $|\mathcal{K}_f|=5$, $|\mathcal{N}_i|=7$ $\forall i \in \mathcal{K}_f$, $L=8$, $\sigma^2=-90.8$ dBm, and $W=100$ MHz.
Also, the BSs of the same type are assumed to have the same maximum transmit power such that: $P_{c,\textrm{max}}=30$ dBm, $P_{m,\textrm{max}}=37$ dBm, $P_{p,\textrm{max}}=27$ dBm, $P_{f,\textrm{max}}=20$ dBm. The RRHs, the BSs, and the CRAN users are deployed uniformly within a $3 \times 3$ km rectangular grid. The HetNet users are deployed uniformly within a circle centered at each BS of radius equal to the coverage of each BS. The coverage radius of the macro, pico, and femto BSs are set to: $r_m=1000$ m, $r_p=150$ m, and $r_f=10$ m respectively. All channel coefficients are generated according to a Rayleigh distribution with variance $10$ dBm. Statistical results are averaged over a large number of independent runs. For comparisons, we consider the NE, the proposed CHE, and an equal power allocation scheme.
\begin{figure}[t]
	\centering
	\includegraphics[width=8 cm,height=5cm,angle=0]{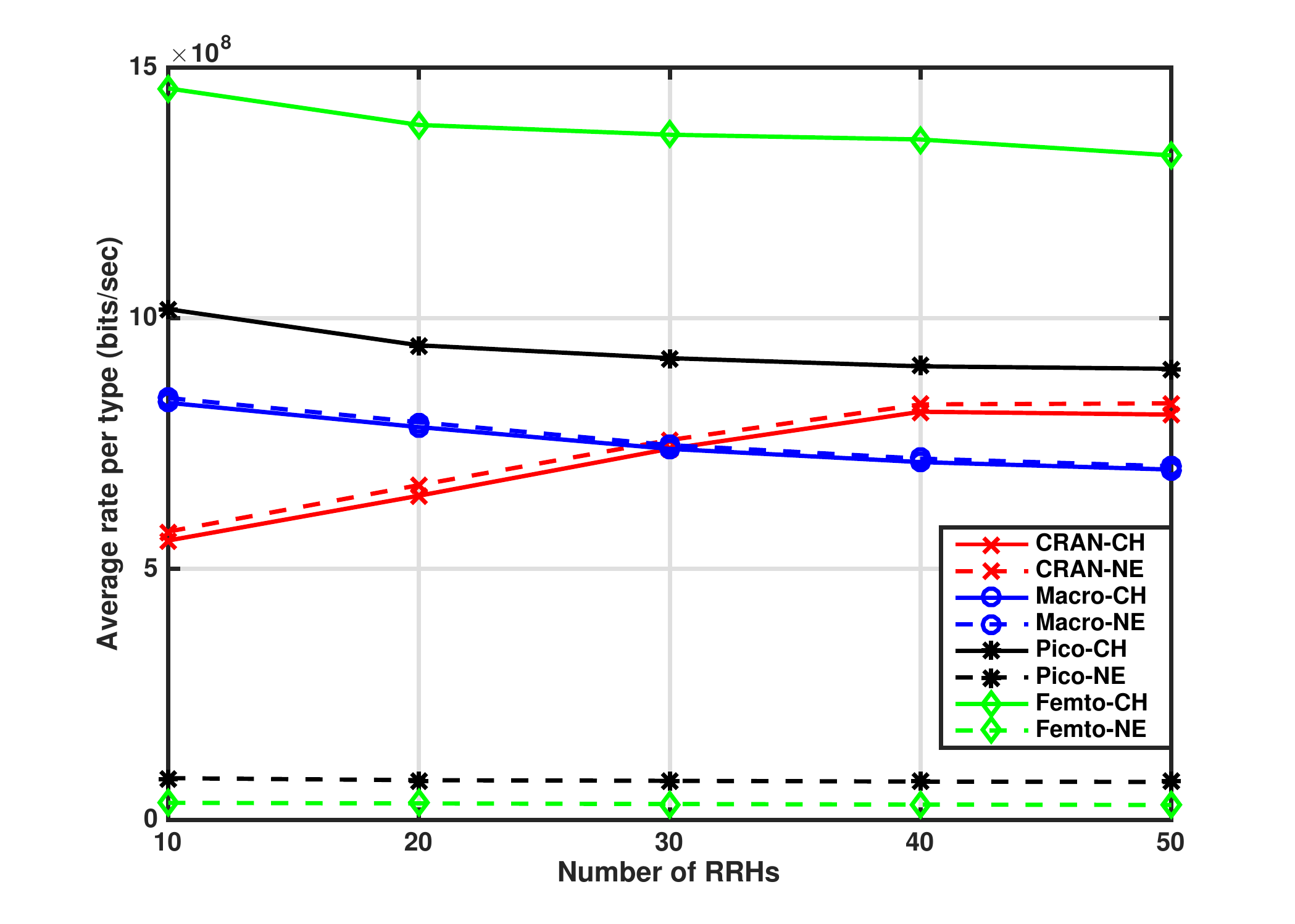}
	\caption{ Average rate per type versus number of RRHs.
	}\vspace{-0.4 cm}\label{RRHrate}
\end{figure}

\begin{figure}[t]
	\centering
	\includegraphics[width=8 cm,height=5cm,angle=0]{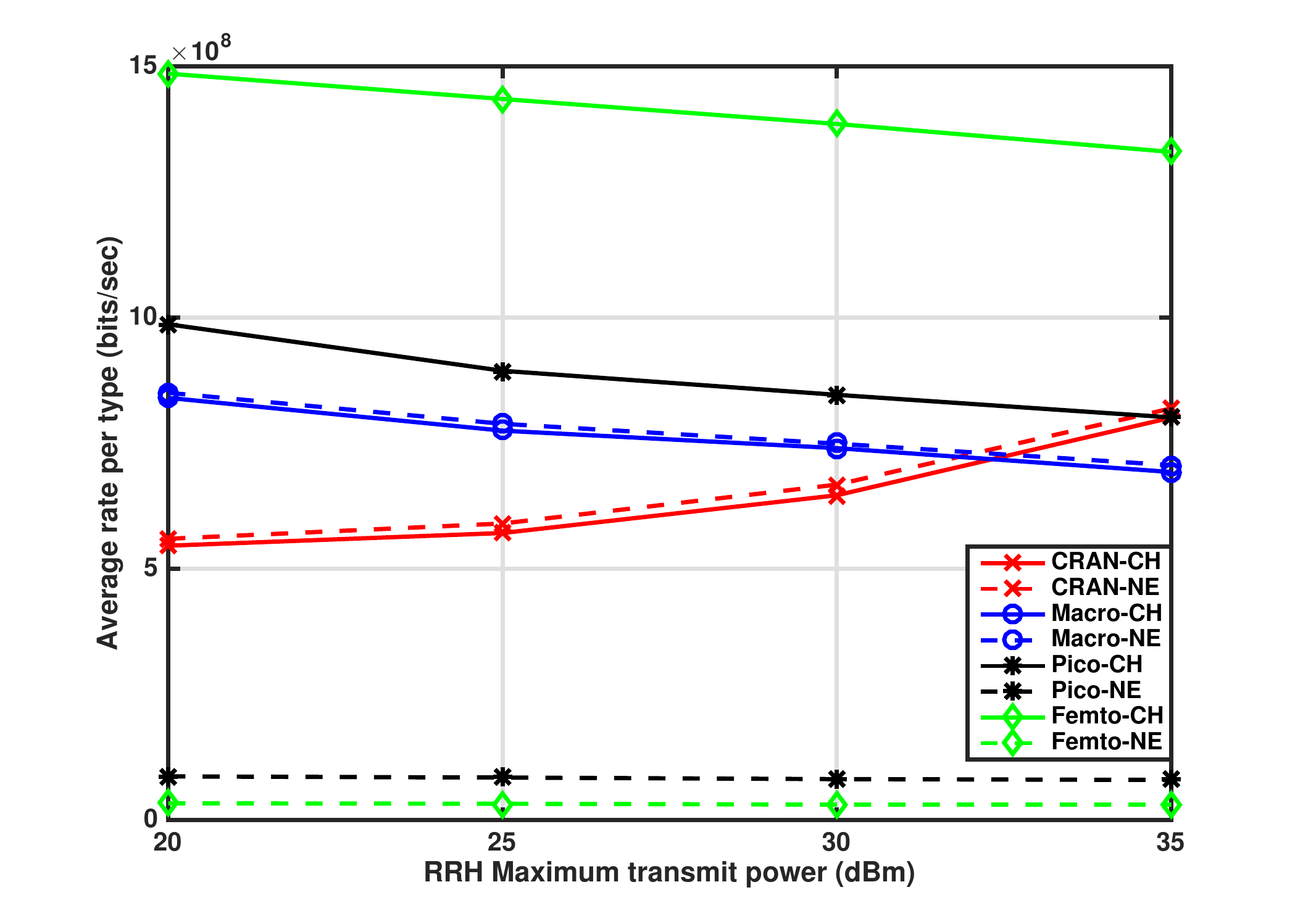}
	\caption{ Average rate per type versus RRH maximum transmit power.
	}\vspace{-0.4 cm}\label{RRHpowrate}
\end{figure}

\begin{figure}[t]
	\centering
	\includegraphics[width=8 cm,height=5cm,angle=0]{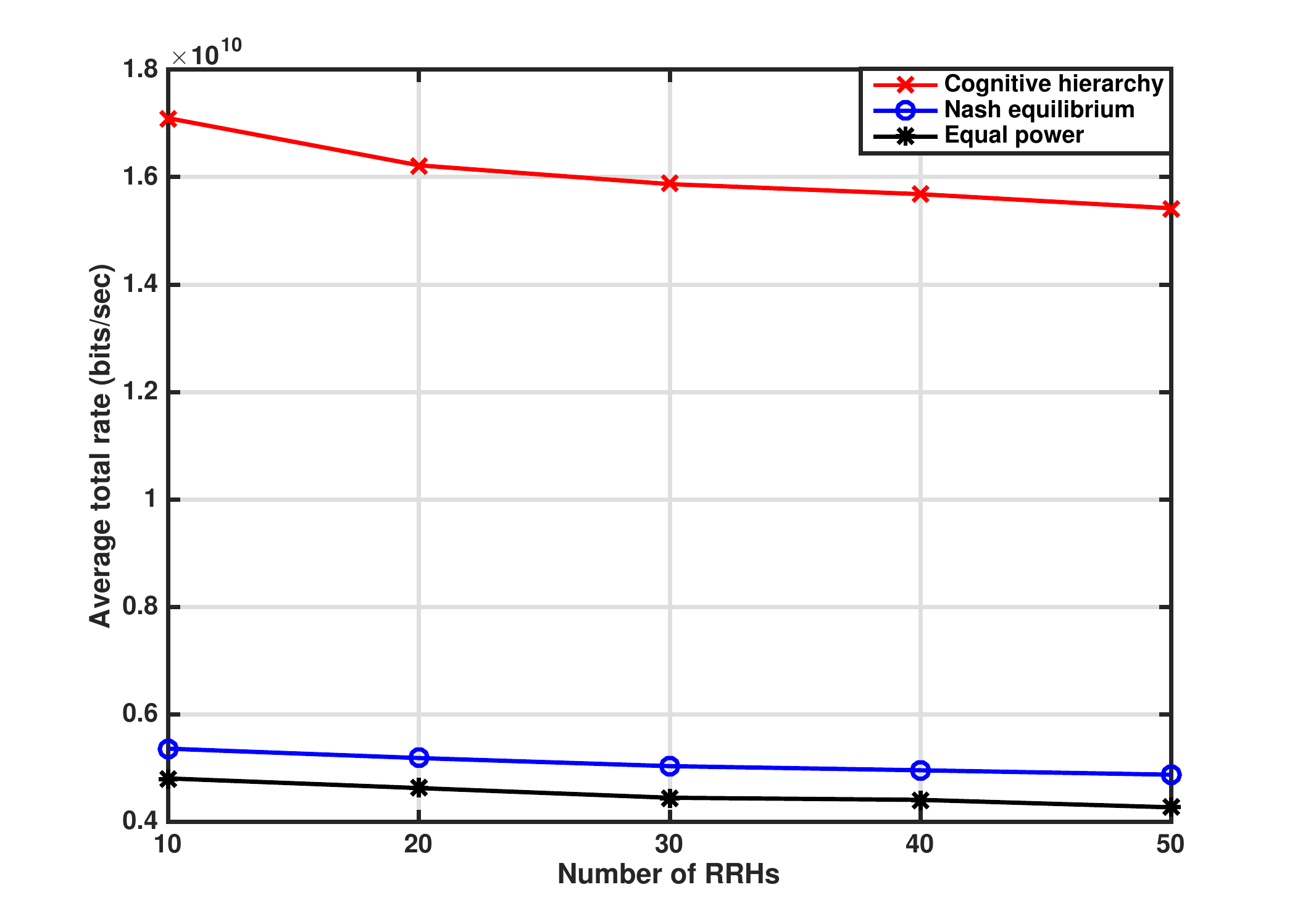}
	\caption{Average total rate versus number of RRHs.
	}\vspace{-0.4 cm}\label{RRHtotal}
\end{figure}

\begin{figure}[t]
	\centering
	\includegraphics[width=8 cm,height=5cm,angle=0]{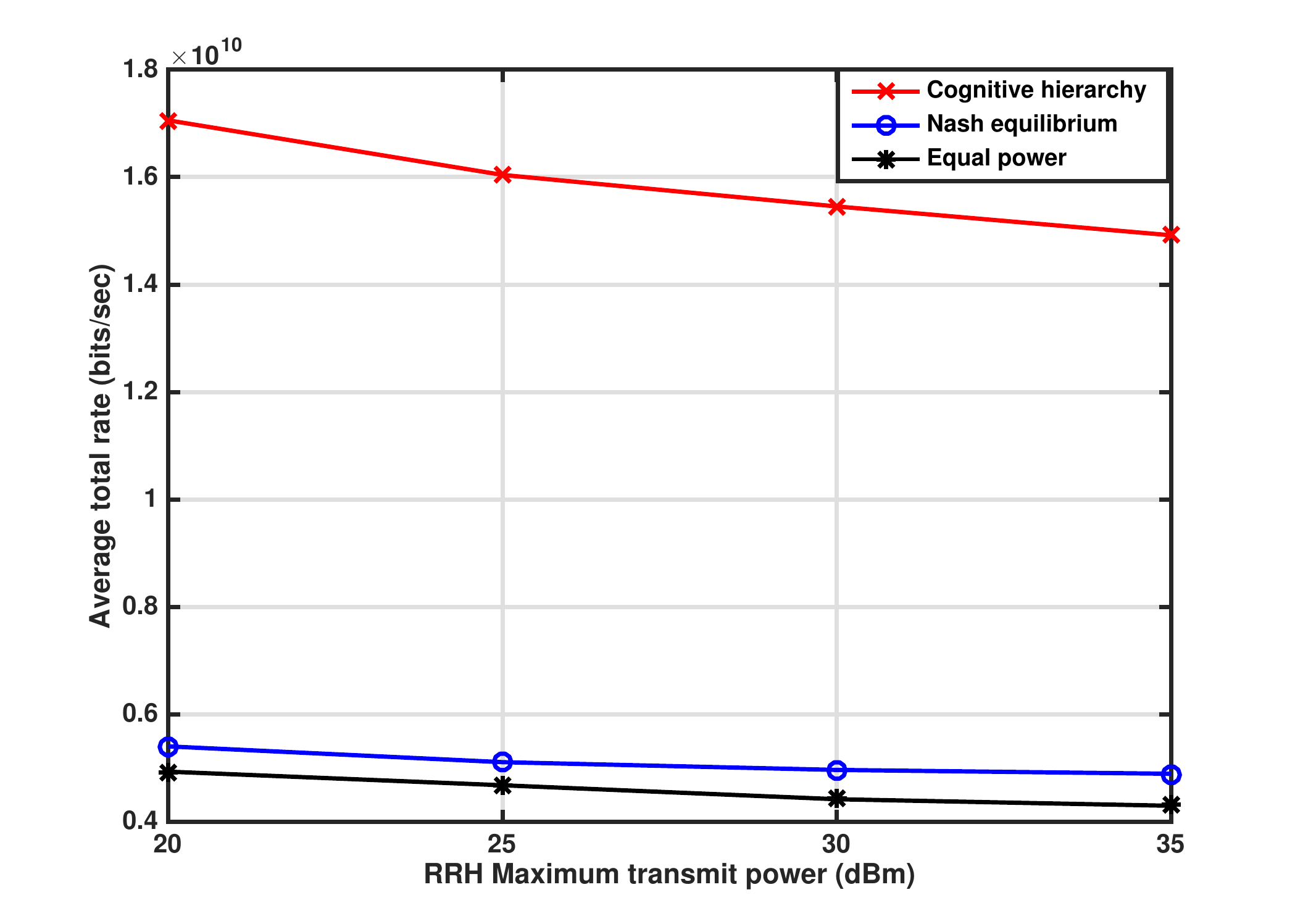}
	\caption{Average total rate versus RRH maximum transmit power. }\vspace{-0.4 cm}\label{RRHpowtotal}
\end{figure}

Fig. \ref{RRHrate} shows that average rate of the CRAN increases with the number of RRH while the average rate of the BSs in $\mathcal{H}$ decreases due to interference.
In this figure, we can see a drastic increase in the sum-rate of femtocells and picocells under CH as opposed to the NE. The average sum-rate of the pico and femto BS increases from $10^7$ bits/sec to $10^8$ and $10^9$ bits/sec, respectively. As for the macro cells and the CRAN, the average sum-rate decreases slightly with CH. The average decrease is only $2.4\%$ for macro BSs and $2.6\%$ for the CRAN. 

Fig. \ref{RRHpowrate} shows the sum-rate of the CRAN varies with the maximum RRH transmit power. Clearly,  the average CRAN sum-rate increases while that of the BSs decreases due to the increased interference. In Fig. 2, we can also see, analogously to Fig. 1, that CH yields significant overall rate improvements, particularly for pico and femto BSs.

Figs. 3 and 4 show the average total system rate as the number of RRHs and their maximum power vary, respectively.
Fig. \ref{RRHtotal} shows a considerable improvement in the total system rate when CH is used compared to both cases of NE and equal power policy. The improvement in the total system rate is on the average double the total sum-rate of the NE and the equal power policy. Fig. \ref{RRHpowtotal} further corroborates such rate improvements at all RRH transmit power values. Clearly, the proposed CH power allocation approach improves considerably the performance of the low-powered (pico and femto) BSs without degrading the performance of the CRAN and the macro BSs. It also results in a considerable improvement in the overall system performance.


\section{Conclusion} 
In this paper, we have studied the problem of power control in a hybrid cellular network in which a CRAN and a small cell heterogeneous network co-exist. We have formulated the problem as a noncooperative game and we have studied the solution under a conventional Nash equilibrium and under a case in which the base stations are grouped in hierarchies according to the framework of cognitive hierarchy theory. Simulation results have shown a considerable system performance in terms of the total sum-rate compared to conventional solutions. The results have also shown that the proposed CH based approach reduces considerably the interference on the low-powered pico and femto cells without jeopardizing the performance of macro base stations and the CRAN.



\def\baselinestretch{0.9}


\begin{thebibliography}{1}
	
\bibitem{cloudbenefits2}A. Checko, H. Christiansen, Y. Yan, L. Scolari,
G. Kardaras, M. Berger, and L. Dittmann, ``Cloud RAN for mobile networks—a technology overview,'' \emph{IEEE Communication Surveys and Tutorials}, vol.\ 17, no.\ 1, pp.~ 405-426, March 2015.
	

\bibitem{hetcloud} M. Peng, Y. Li, J. Jiang, J. Li, and C. Wang “Heterogeneous cloud radio access networks: a new perspective for enhancing
spectral and energy efficiencies,” \emph{IEEE Wireless Commun.}, vol.\ 21, no.\ 6, pp.~ 126–135, Dec. 2014.

\bibitem{hetint1}T. Ran, S. Sun, B. Rong, and M. Kadoch,``Game theory based multi-tier spectrum sharing for LTE-A heterogeneous networks,'' \emph{in Proc. of IEEE International Conference on Communications (ICC)}, London, UK, June 2015.

\bibitem{hetint2}S. Cheng, S. Lien, F. Chu, and K. Chen,``On exploiting cognitive radio to mitigate interference in macro/femto heterogeneous networks,'' \emph{IEEE Wireless Communications}, vol.\ 18, no.\ 3, pp.~ 40 - 47, June 2011.

\bibitem{hetint3}C. Shen, Jie Xu, and M. van der Schaar``, Silence is gold: Strategic interference mitigation using tokens in heterogeneous small cell networks,'' \emph{IEEE Journal on Selected Areas in Communications}, vol.\ 33, no.\ 6, pp.~1097-1111, June 2015.

\bibitem{cloud0} H. Zhang, C. Jiang, and J. Cheng, ``Cooperative interference mitigation and handover management for heterogeneous cloud small cell networks,'' \emph{IEEE Wireless Communications}, vol.\ 22, no.\ 3, pp.~ 92-99, June 2015.

\bibitem{cloud1}L. Liu, S. Bi, and R. Zhang, ``Joint power Control and fronthaul rate allocation for throughput maximization in OFDMA-based cloud radio access network,'' \emph{IEEE Trans. on Communications}, vol.\ 63, no.\ 11, pp.~ 4097-4110, November 2015.

\bibitem{cloud2} M. Peng, K. Zhang, J. Jiang, J. Wang, and
W. Wang, ``Energy-Efficient resource assignment and power
allocation in heterogeneous cloud radio
access networks,'' \emph{IEEE Trans. on Vehicular Technology}, vol.\ 64, no.\ 11, pp.~ 5275-5287, November 2015.

\bibitem{miso} M. Vu ``MISO capacity with per-antenna power constraint,'' \emph{IEEE Trans. on Communications}, Vol.\ 59, No.\ 5, pp.~ 1268-1274, May 2011. 

\bibitem{waterfilling} A. Goldsmith, \emph{Wireless Communications}, Cambridge University Press, 2005.

\bibitem{subcarrier} L. Yanhui, W. Chunming, Y. Changchuan, and Y. Guangxin, ``Downlink scheduling and radio resource allocation in adaptive OFDMA wireless communication systems for user-individual QoS, '' \emph{In Proceedings of the World Academy of science, engineering and technology}, vol.\ 12, pp.~221-225, March 2006.

\bibitem{Camerer_acognitive} C. F. Camerer and J. Chong, ``A cognitive hierarchy model of games,'' \emph{Quarterly Journal of Economics}, vol.\ 119, no.\ 3, pp.~ 861-898, 2004.

\bibitem{stackelberg1} S. Guruacharya, D. Niyato, D. In Kim, and E. Hossain, ``Hierarchical competition for downlink power
allocation in OFDMA femtocell networks,'' \emph{IEEE Trans. on Wireless Commun.}, vol.\ 12, no.\ 4, pp.~1543-1553, April 2014.

\bibitem{gametheory} Z. Han, D. Niyato, W. Saad, T. Ba\c{s}ar, and Are Hj{{\o{}}}rungnes, \emph{Game Theory in Wireless and Communication Networks: Theory, Models, and Applications}, Cambridge University Press, 2012.

\bibitem{purenash} G. Bacci, S. Lasaulce, W. Saad, and L. Sanguinetti,''Game Theory for Networks: A tutorial on game-theoretic tools for emerging signal processing applications,`` IEEE Signal Processing Magazine, vol.\ 33, no.\ 1, pp.~ 94-119, January 2016.




	

\end{thebibliography}
\end{document}